\documentclass[a4paper,UKenglish,cleveref, autoref, thm-restate]{lipics-v2021}
\usepackage[export]{adjustbox} %

\usepackage{amssymb,amsmath, amsthm}
\usepackage{cleveref}
\usepackage{graphicx}
\usepackage{xcolor}
\usepackage{wasysym}

	\newcommand*{\myproblem}[1]{\noindent\underline{\textsc{#1}}.~}

\definecolor{solarizedYellow}{HTML}{B58900}
\definecolor{solarizedOrange}{HTML}{CB4B16}
\definecolor{solarizedRed}{HTML}{DC322F}
\definecolor{solarizedMagenta}{HTML}{D33682}
\definecolor{solarizedViolet}{HTML}{6C71C4}
\definecolor{solarizedBlue}{HTML}{268BD2}
\definecolor{solarizedCyan}{HTML}{2AA198}
\definecolor{solarizedGreen}{HTML}{859900}

\newcommand*{\locus}{\mathsf{locus}}
\newcommand*{\strlabel}{\mathsf{label}}
\newcommand*{\pred}{\mathsf{pred}}
\newcommand*{\suc}{\mathsf{succ}}
\newcommand*{\pali}[1]{\overleftarrow{#1}} %
\newcommand*{\STree}{\mathsf{ST}}
\newcommand*{\lca}{\mathsf{lca}}
\newcommand*{\timeSU}{\ensuremath{\log^2\log n}}
\newcommand*{\colors}{\mathcal{C}}

\newcommand*{\fact}{\mathcal{F}}
\newcommand*{\rext}{E_r}
\newcommand*{\sext}{S_r}

\usepackage{xcolor}
\newcommand{\gk}[1]{\textcolor{blue}{#1}}
\newcommand{\dk}[1]{\textcolor{red}{#1}}
\newcommand{\hili}[1]{\textcolor{teal}{#1}}

\renewcommand{\gk}[1]{}
\renewcommand{\dk}[1]{}
\renewcommand{\hili}[1]{#1}

\nolinenumbers

\author{Dominik {Köppl}}{University of Yamanashi, Kofu, Japan \and \url{https://dkppl.de}}{dkppl@dkppl.de}{0000-0002-8721-4444}{This work was supported in part by JSPS KAKENHI Grant Number 25K21150. We thank Fédération de Recherche Bézout for supporting the research visit of the first author to LIGM, where this work was initiated.}
\author{Gregory Kucherov}{LIGM, CNRS and Gustave Eiffel University, Marne-la-Vall\'ee, France \and \url{http://igm.univ-mlv.fr/~koutcher/}}{gregory.kucherov@univ-eiffel.fr}{0000-0001-5899-5424}{}

\authorrunning{Köppl and Kucherov}
\Copyright{D. Köppl and G. Kucherov}

\title{Smallest suffixient set maintenance in near-real-time}
\titlerunning{Smallest suffixient set maintenance in near-real-time}


\begin{document}
\maketitle
	
	\begin{abstract}
		The size of the \textit{smallest suffixient set} of positions of a string recently emerged as a new measure of string \textit{repetitiveness} -- a measure reflecting how much of repetitive content 
		the string contains. We study how to maintain the smallest suffixient set online in near-real-time, that is with small (in our case, polyloglog) worst-case time for processing each letter. Two frameworks are considered: 
		when the text is given letter-by-letter in either right-to-left or left-to-right order. 
		Our central algorithmic tool is  Weiner's suffix tree algorithm and associated algorithmic primitives for its efficient implementation. 
	\end{abstract}
	
	\vspace*{\fill}
	\clearpage
	\setcounter{page}{1}
	\section{Introduction}
	
	Nowadays, it is not uncommon to deal with highly repetitive data in the sense that many data fragments are not locally consecutively repeated but rather spread across the entire data.
	Examples include genomic sequences, files of version control systems or archives, sensor data, and many others. 
	This has motivated recent work on how indexing data structures can exploit this redundancy~\cite{navarro21indexing1,navarro21indexing2}. A formal framework for such studies requires a definition of a measure of redundancy (repetitiveness) of data, and there are many ways to define  such a measure \cite{navarro21indexing1}. These measures can be classified into two categories: those based on compressed representations of the data and those based on ``extrinsic'' properties of the data, unrelated to any compressed representation. For sequential data, the first category can be exemplified by dictionary-based compression algorithms, such as Lempel--Ziv or grammar-based compression. As for the second category, several new measures have been defined, such as \textit{smallest string attractor} \cite{kempa18stringattractors} or \textit{normalized substring complexity} \cite{raskhodnikova13sublinear,DBLP:journals/corr/abs-2510-16454}. 

	A string attractor is a set of positions in the string such that
	every \textit{distinct} substring has an occurrence that covers one of those positions. 
	The size of the smallest attractor $\gamma$ of a string of length $n$ is a lower bound on the number of phrases of the Lempel--Ziv parsing or the size of a grammar producing this string, 
	but, on the other hand, those can be bounded as $O(\gamma\log\frac{n}{\gamma})$, which shows that $\gamma$ is a meaningful compressibility estimator. It has also been shown that space $O(\gamma\log\frac{n}{\gamma})$ is sufficient for a data structure capable of supporting string matching in a time linear in the pattern length  \cite{10.1145/3426473}. 
	
	Unfortunately, computing the smallest attractor is an NP-complete problem \cite{kempa18stringattractors}. Recently, a related measure was proposed:  
	the size of the \textit{smallest suffixient set}, denoted $\chi$ \cite{DBLP:journals/corr/abs-2312-01359}. The smallest suffixient set is an attractor that is not necessarily a smallest string attractor~\cite{DBLP:journals/corr/abs-2312-01359}.  %
	In contrast to $\gamma$, $\chi$ can be computed in linear time \cite{DBLP:conf/spire/CenzatoOP24}. Testing if a given set of positions is suffixient is done in linear time as well  \cite{DBLP:journals/corr/abs-2506-08225}. 
	Several properties of $\chi$ have been established recently. These include relationships of $\chi$ with other repetitiveness measures and how it can be affected by transformations of the input string, such as edits or reversal \cite{DBLP:conf/spire/NavarroRU25}. Very recently, it has been shown that a string can be encoded into a lossless representation of size $O(\chi)$ \cite{ShibataBannai26}. 
	
	In this paper, we study how to compute $\chi$ in the \textit{online} mode in \textit{near-real-time}, that is with small (in our case, double-logarithmic) worst-case time guarantees per letter. 
	An online algorithm for computing $\chi$ has already been proposed by Navarro et al.~\cite{DBLP:journals/corr/abs-2506-05638}, which, however, provides only \textit{amortized} time bounds, whereas our primary interest here is in worst-case bounds. 
	
	Since the definition of suffixient set is ``oriented'', i.e. is not invariant under string reversal (as opposed, e.g., to string attractors), for online algorithms it is relevant whether the string is input left-to-right or right-to-left. Both cases were considered in \cite{DBLP:journals/corr/abs-2506-05638}. For the left-to-right input, we propose a different computation based on Weiner's suffix tree algorithm, following the techniques that we recently applied to some other problems on strings \cite{KopplKucherov2026}. For the right-to-left input, we propose a 
	version of the algorithm of \cite{DBLP:journals/corr/abs-2506-05638} admitting an efficient near-real-time implementation. 
	
	Since Weiner's algorithm and its efficient implementations are central to our work,  Section~\ref{sec:algotools} below is devoted to a short presentation of this algorithm, its improvements and associated algorithmic primitives. Our main results are then presented in Section~\ref{sec:results}. 
	
	\section{Supermaximal extensions and suffixient sets}
	We consider an integer alphabet $\Sigma$ and assume $\sigma=|\Sigma|=n^{O(1)}$ where $n$ is the input size. 
	Consider a text $w\in\Sigma^*$ and let $\fact_w$ be the set of substrings of $w$. A substring $u$ is called \textit{right-maximal} if for two distinct letters $x\neq y$, we have $ux,uy\in\fact_w$. 
	Given a right-maximal substring $u$, a substring $ux\in\fact_w$ is called a \textit{right extension}. In what follows it will be convenient for us to represent a right extension as a pair $(u,x)$ where $u\in\Sigma^*$ and $x\in\Sigma$. 
	Let $\rext(w)$ denote the set of all right extensions of $w$. 
	
	Consider $\rext(w)$ and consider its minimal subset $\sext(w)\subseteq \rext(w)$ 
	such that for every $(u,x) \in \rext(w)$, there exists $(v,x) \in \sext(w)$ such that $u$ is a suffix of $v$. 
	Due to minimality, for every pair $(u,x),(v,x)\in \sext(w)$ with $u \neq v$, neither $u$ is a suffix of $v$ nor $v$ is a suffix of $u$. 
		The subset $\sext(w)$ is uniquely defined and contains all those right extensions $(u,x)$ for which $(zu,x)$ is not a valid right extension for any letter $z\in\Sigma$. 
	The latter means that either $zux\notin\fact_w$ or $zux\in\fact_w$ but $zu$ is not right-maximal, i.e. $zuy\notin\fact_w$ for any other letter $y\neq x$. 
	
	Elements of $\sext(w)$ are called \textit{supermaximal (right-)extensions} (hereafter, SREs) \cite{DBLP:journals/corr/abs-2407-18753,DBLP:conf/spire/NavarroRU25}. By definition, every right extension of $w$ is a suffix of a right extension of $\sext(w)$. Since $\sext(w)$ is uniquely defined, it is also the smallest set of right extensions with this property. 
	
	Each SRE $(u,x)\in\sext(w)$ can be mapped to the end position of an arbitrary occurrence of $ux$ in $w$. 
	Distinct supermaximal extensions are necessarily mapped to distinct positions~\cite[Lemma~1]{DBLP:conf/spire/CenzatoOP24}.
	Therefore, any such mapping defines a one-to-one correspondence between SREs and a certain subset of positions of $w$. Such a set of positions is called a \textit{smallest suffixient set} (hereafter, SSS), 
	i.e.\ a smallest set of positions such that every right extension $(u,x)$ has an occurrence of $ux$ ending at one of those positions. 
	Consequently, the set of SREs (based on substring distinctness) is uniquely defined whereas an SSS (based on the occurrences of SREs) is not. 
		
	\begin{example}\label{ex:suffixient_example}
	For the string $w=\mathtt{aabaababa\$}$, $\sext(w)$ consists of three SREs: $(\texttt{aaba},\texttt{a})$, $(\texttt{aaba},\texttt{b})$ and $(\texttt{aba},\texttt{\$})$. Its unique SSS consists of positions shown underlined in the string:
	$\mathtt{aaba}\underline{\mathtt{a}}\mathtt{ba}\underline{\mathtt{b}}\mathtt{a}\underline{\mathtt{\$}}$. For $v=\mathtt{aabaabaab\$}$, $\sext(v)=\{(\mathtt{a},\mathtt{a}), (\mathtt{a},\mathtt{b}),(\mathtt{aabaab},\mathtt{a}),(\mathtt{aabaab},\mathtt{\$})\}$.  It has several (nine) possible SSSs, such as $\mathtt{a\underline{a}\underline{b}aab\underline{a}ab\underline{\$}}$, $\mathtt{a\underline{a}baa\underline{b}\underline{a}ab\underline{\$}}$, etc. 
	\end{example}

	The size of an SSS, or equivalently, the size of $\sext(w)$, denoted $\chi$, is a measure of repetitiveness~\cite{navarro21indexing1} of a string, as presented in the introduction. Whereas $\chi$ is incomparable with common ``copy-paste measures'' such as measures based on Lempel--Ziv compression or context-free grammar representation,  in the sense that it can happen to be asymptotically smaller or asymptotically larger than those \cite{navarro21indexing1}, it is still a meaningful measure of repetitiveness. 
	In particular, we have $\gamma\leq \chi$ on the one hand, and $\chi\leq 2r$ on the other hand, where $r$ is the repetitiveness measure based on Burrows--Wheeler transform \cite{DBLP:conf/spire/NavarroRU25}. 
	
	For example, some remarkable words studied in word combinatorics, such as Fibonacci words, have $\chi=O(1)$, which reflects their high repetitiveness \cite{DBLP:conf/spire/NavarroRU25}. 
	On the other hand, for a de Bruijn sequence of length $n$, we have $\chi=\Theta(n)$, which shows a difference with copy-paste measures, which are all bounded by $O(n/\log n)$. 
	
	\section{Algorithmic tools}
	\label{sec:algotools}
	\subsection{Weiner's suffix tree algorithm: an outline}
	\label{sec:weiner}
	We assume the reader to be familiar with suffix trees, see e.g. \cite{gusfield97algorithms}. 	
	Consider the suffix tree for a given text $w$. 
	Edges are labeled by substrings of $w$, which we call \textit{edge labels}.
	We call an edge an $x$-edge, for $x\in\Sigma$, if its edge label starts with $x$. 
	Each substring $u\in\fact_w$ corresponds to a \textit{locus} in the suffix tree, denoted $\locus(u)$, which can be found by spelling $u$ starting from the root of the tree. $\locus(u)$ is either an internal {node} or a leaf of the tree (hereafter called simply a \textit{node}), or an \textit{implicit node} specified by an offset from the parent node of some edge. 
	The substring whose locus is a node $\alpha$ is called its \textit{label} and is denoted $\strlabel(\alpha)$. 
	
	Weiner's algorithm constructs the suffix tree for a text in a ``reversed online'' fashion, i.e. by processing the text right-to-left and after reading a letter of the text, updating the suffix tree by inserting a new suffix of the text. We will rely on Weiner's algorithm in this work. Here we focus only on the features of Weiner's algorithm relevant to the present work, for its full description we refer the reader to \cite{breslauer13near,KopplKucherov2026}. 
	
	Weiner's algorithm maintains a suffix tree augmented with Weiner's links, or \textit{W-links} for short. W-links are pointers indexed by alphabet letters: a node $\alpha$ has a defined W-link by a letter $x\in\Sigma$ if $x \cdot \strlabel(\alpha)$ is a substring of $w$. We call it an $x$-link and define $W_x(\alpha)$ to be the destination of the $x$-link from $\alpha$. W-links are used by Weiner's algorithm to jump from $\alpha$ to $\gamma=\locus(x \cdot \strlabel(\alpha))$.  If $\gamma$ is a node of the tree, then $W_x(\alpha)=\gamma$ and the link is called \textit{hard}. If $\gamma$ is an implicit node, the link is called \textit{soft} and $W_x(\alpha)$ is defined to be the closest descendant node\footnote{This definition of soft links does not match the original paper \cite{weiner73linear} and has been introduced later in work~\cite{breslauer13near} that we comment on in the next section. We use this definition in our paper as well.} of $\gamma$. 
	Figure~\ref{fig:suffixient_example} in \cref{app:figs} shows the suffix trees for the two strings from Example~\ref{ex:suffixient_example}. It also shows which W-links are defined for each node and the loci of the SREs.
	
	Consider a string $w$ ending with a unique sentinel symbol $\$$. The following lemma summarizes properties of W-links that will be useful to us. 
	\begin{lemma}
		\label{lem:SREreversed}
		The following statements hold.
		\begin{enumerate} [(i)]
			\item Every node has at least one defined W-link except for the leaf node labeled by the entire string $w$. \label{case:SREreversed1}
			\item For any $x\in\Sigma$, if an $x$-link is defined for a node $\nu$, then it is defined for any ancestor of $\nu$. \label{case:SREreversed2}
			\item If a node $\nu$ has a defined soft $x$-link, there exists a descendant node of $\nu$ with a defined hard $x$-link. 
			\label{case:SREreversed3}
		\end{enumerate}
	\end{lemma}
	\begin{proof}
		\begin{enumerate} [(i)]
			\item The label of every internal node has at least two occurrences in $w$ and therefore at least one preceded by some letter. Every leaf is labeled by a suffix that is preceded by some letter as well, with the only exception of the leaf labeled by the entire string. 
			\item The claim follows from the observation that if a substring has an occurrence preceded by a letter $x$, then this holds for any of its prefixes as well. 
			\item Due to the unique terminal symbol $\$$, the locus of any suffix of $w$ is a leaf, and therefore any leaf except for the one labeled by the entire string $w$ has a single hard W-link pointing to another leaf. Assume now an internal node $\nu$ has a soft $x$-link. Its label $u=\strlabel(\nu)$ must extend to a suffix preceded by $x$. The locus of this suffix is a leaf which is a descendant of $\nu$ and has a defined hard $x$-link. 
		\end{enumerate}
	\end{proof}

	We now outline a round of Weiner's algorithm transforming the suffix tree for a text $w$ to the one for $xw$ for some $x\in\Sigma$. We refer to Figure~\ref{fig:weiner_suffix_update} in our description. The algorithm maintains the leaf $\lambda$ whose label is the entire string $w$. A key step of the algorithm is to locate the node $\alpha$ that is the lowest ancestor of $\lambda$ with a defined $x$-link. That is, $\strlabel(\alpha)$ is the longest prefix of $w$ that has another occurrence in $w$ preceded by $x$. In the original version \cite{weiner73linear}, this is done by a straightforward upward traversal of the ancestors of $\lambda$ until finding a node with a defined $x$-link. Although on an individual letter this work can take time up to $O(n)$, it takes $O(1)$ time per letter \emph{amortized} over the entire run of the algorithm. 
	
	The goal of locating $\alpha$ is to use its $x$-link in order to locate the node $\gamma=W_x(\alpha)$ to which a new leaf $\lambda'$ should be attached, such that $\lambda'$ has the string label $xw$. 
	Here two cases occur, depending on whether $\alpha$'s $x$-link is hard or soft, and, respectively, whether (a) node $\gamma$ is an existing node or (b) an implicit one (i.e., a locus on an edge). 
	In the latter case~(b), $\gamma$ is created by ``splitting'' the corresponding edge (see Figure~\ref{fig:weiner_suffix_update}). In both cases~(a) and~(b), a new leaf of $\gamma$ is created labeled by $xw$ and the new edge is a $y$-edge where $y$ is the letter following the prefix occurrence of $\strlabel(\alpha)$ in $w$. 
	
	On top of this, we have to update certain W-links. 
	We create a \textit{hard} $x$-link at $\lambda$ and, if necessary, harden the $x$-link of $\alpha$, spending constant time. 
	Furthermore, we have to set or modify a non-constant number of \textit{soft} W-links that are divided into two groups. 
	One group consists of $O(\sigma)$ W-links that should be assigned to the newly created node $\gamma$ by copying the W-links of its child node $\beta$. 
	The other group consists of new $x$-links at nodes on the path between $\lambda$ and $\alpha$ and modified existing $x$-links at nodes on the path between $\alpha$ and $\delta$. %

		\begin{figure}[t]
		\centering
			\includegraphics[width=0.45\textwidth,page=1]{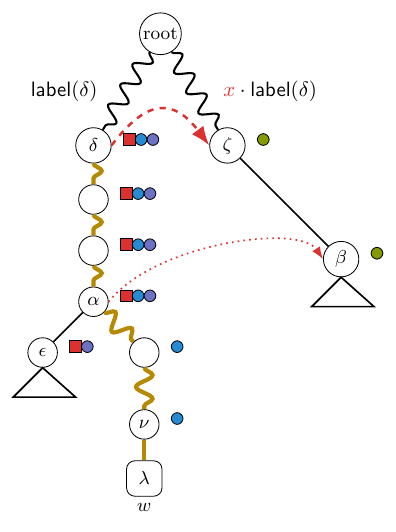}
			\hspace{2em}
			\includegraphics[width=0.45\textwidth,page=2]{suffixientimg/suffixupdate.pdf}
		\caption{%
    One round of Weiner's suffix tree construction algorithm: updating $\STree(w)$ (left) to $\STree(xw)$ (right) by inserting the new suffix $xw$.
			W-links are visualized by colored boxes and circles next to nodes.
			A red box ({\color{solarizedRed}$\blacksquare$}) at a node represents an $x$-link.
			Blue ({\color{solarizedBlue}$\CIRCLE$}), violet ({\color{solarizedViolet}$\CIRCLE$}) and green ({\color{solarizedGreen}$\CIRCLE$}) circles represent \textit{sets} of W-links that can, in general, be growing towards the root. Note also that blue and violet sets are not necessarily disjoint. 
			Explicitly drawn are some $x$-links: Dashed red arrows represent hard $x$-links, dotted red arrows represent soft $x$-links.
    Curly edges represent paths that may contain multiple nodes.
    Node $\delta$ is the lowest ancestor of $\lambda$ with a \emph{hard} $x$-link, and node $\alpha$ is the lowest ancestor of $\lambda$ with a (not necessarily hard) $x$-link. 
    $\alpha$'s $x$-link can be soft (as in the figure) or hard (in case $\alpha=\delta$). 
    If this link is soft, the algorithm creates a new node $\gamma$ that is the insertion point for the leaf~$\lambda'$. 
    If this link is hard, the insertion point is $W_x(\alpha)$. 
    After creating $\gamma$, the W-links on the golden %
    thick curly path from $\lambda$ to $\delta$ need to be updated.
	}
		\label{fig:weiner_suffix_update}
	\end{figure}
	
	\subsection{Fringe colored ancestor problem}
	\label{sec:fringe}
	
	In their seminal paper \cite{breslauer13near}, Breslauer and Italiano proposed an implementation of Weiner's algorithm working in \textit{near-real-time}, that is providing a small (double logarithmic) worst-case time bound on an individual round. 
	Their main idea is to speed up the retrieval of $\alpha$ from $\lambda$ by answering a query called \textsc{fringe marked ancestor} on dynamic trees. 
	By Lemma~\ref{lem:SREreversed}(\ref{case:SREreversed2}), nodes with a defined $x$-link, for a given $x\in\Sigma$, form a connected subtree of the suffix tree sharing the same root; 
	the leaves of this subtree are called the \emph{fringe}.
	Retrieving $\alpha$ 
	amounts to computing the lowest ancestor of the fringe. 
	We define a more general version of this problem where we simultaneously store information about W-links for all letters (expressed by \emph{colors}\footnote{We introduce the term \emph{colors} here as a means of abstraction since they will also play another role, besides being synonymous to letters, later in \cref{sec:maintainingSREs}.}), as opposed to a separate data structure for each letter (as in \cite{breslauer13near}). 
	We call this problem \textsc{fringe colored ancestor} and define it as follows.
	
	\myproblem{Fringe colored ancestor}
	Consider a set of colors $\colors$ and a dynamic tree initially consisting of a single uncolored node. Maintain the tree under the following operations:\footnote{The set of operations can be larger and include deletion and uncoloring of a node; here we only consider operations that we need for our purposes.}
	\begin{itemize}
		\item insert a new uncolored child of an existing node,
		\item assign a color $x\in\colors$ to a node $\nu$ provided that either $\nu$ is the tree root or the parent of $\nu$ is colored with $x$,
		\item introduce a new node by splitting an edge and coloring this node by the colors assigned to its descendant,
		\item for a given node and a given color $x\in\colors$, find its lowest ancestor colored with $x$. 
	\end{itemize}
	In particular, \textsc{fringe marked ancestor} is \textsc{fringe colored ancestor} with a single color. 
	The updates preserve the \emph{fringe property}, i.e., ancestors of a node colored by $x$ are colored by $x$ as well. 
	The fringe property is important for modelling W-links since they also obey the fringe property by Lemma~\ref{lem:SREreversed}(\ref{case:SREreversed2}).
	
	Generalizing \cite{breslauer13near}, \textsc{fringe colored ancestor} can be solved by reducing it to \textsc{colored predecessor} in dynamic lists:

	\myproblem{Colored predecessor}
	Maintain a dynamic linked list $\mathcal{L}$ where each element is assigned a \textit{color} from a set $C$ of possible colors. Possible updates to $\mathcal{L}$ are insertions of an element after a given element. %
	A colored predecessor query $\mathcal{L}.\pred(\nu,x)$, for an element $\nu \in \mathcal{L}$ of the list and a color $x\in C$, 
	asks for the closest preceding element of $\nu$ in $\mathcal{L}$ colored by $x$. 
	The colored successor $L.\suc(\nu,x)$ is defined symmetrically.

	To explain the reduction from \textsc{fringe colored ancestor} to \textsc{colored predecessor}, we need another operation: 
	the \emph{lowest common ancestor} of two nodes $\mu$ and $\nu$ of a tree, denoted $\lca(\mu,\nu)$. 
	It is known that lowest common ancestor queries can be answered in constant time even on a dynamic tree undergoing insertion and deletion of nodes~\cite{cole05dynamic}. 
	
	We further consider that a node colored by $x\in C$ can be either \emph{explicitly $x$-colored} or  \emph{implicitly $x$-colored}, but we require that a node is implicitly $x$-colored if and only if at least one of its (not necessarily immediate) descendants is explicitly $x$-colored. We call a node \textit{$x$-colored} if it is either explicitly or implicitly $x$-colored. 
	Observe that $x$-colored nodes fulfill the fringe property and each $x$-colored node that has no $x$-colored descendants must be explicitly $x$-colored. 
	
	Consider the linked list~$\mathcal{L}$ storing the  \textit{Euler tour} of the current tree, 
	where each node of the tree has two copies in the list corresponding to its first and last visits. 
	Thus, each element of the list is a node~$\nu$, which we augment by the set of the \textit{explicit colorings} of $\nu$. 
	Then the following holds --- see also Figure~\ref{fig:KN} for an illustration.

		\begin{figure}[h]
		\centering
		\includegraphics[width=0.3\textwidth,page=1]{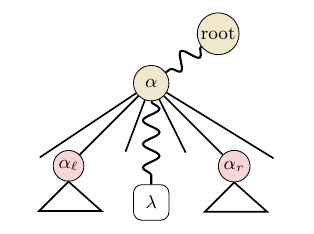}
		\fbox{\includegraphics[scale=0.6,page=2]{suffixientimg/explicitcolor.pdf}}
		\includegraphics[width=0.3\textwidth,page=3]{suffixientimg/explicitcolor.pdf}
		\fbox{\includegraphics[scale=0.6,page=4]{suffixientimg/explicitcolor.pdf}}
		\caption{Illustration to Lemma~\ref{lem:KN}. Explicitly colored nodes are shown in red ({\color{solarizedRed!20}$\blacksquare$}) and implicitly colored nodes are shown in yellow ({\color{solarizedYellow!20}$\blacksquare$}). The Euler tour list~$\mathcal{L}$ is shown next to the tree in top-down order. The tree on the left shows the case when $\lambda$ is between two descendants $\alpha_\ell$ and $\alpha_r$ of $\alpha$ in the Euler tour, and the tree on the right shows the case when all successors of $\lambda$ inside the subtree rooted at $\alpha$ are not explicitly colored (the case that all predecessors of $\lambda$ inside the subtree rooted at $\alpha$ are not explicitly colored is symmetric).
	}
		\label{fig:KN}
	\end{figure}

	\begin{lemma}[{\cite[Lemma 5]{kucherov17fullfledged}}]
		Let $\lambda$ be a node that is not $x$-colored and let $\alpha$ be the lowest $x$-colored ancestor of $\lambda$. Consider the Euler tour list~$\mathcal{L}$ in which only explicit colorings are represented, and let $\alpha_\ell =\mathcal{L}.\pred(\lambda,x)$ and $\alpha_r = \mathcal{L}.\suc(\lambda,x)$. Then $\alpha$ is the lowest node between $\lca(\alpha_\ell,\lambda)$ and $\lca(\lambda,\alpha_r)$. 
		\label{lem:KN}
	\end{lemma}
	\begin{proof}
		Observe that since $\lambda$ is not $x$-colored, no descendant of $\lambda$ is $x$-colored either, and therefore successor and predecessor operations are unambiguously defined. 
		Since $\alpha$ is a marked ancestor of $\lambda$, it remains to show that it is the lowest one. This is immediate by contradiction: if there was a lower marked ancestor, its first occurrence in the Euler tour would have been after the occurrence of $\alpha_\ell$ and its second occurrence would have been earlier than the occurrence of $\alpha_r$. 
	\end{proof}

	\hili{Lemma~\ref{lem:KN} allows us to answer fringe colored ancestor queries by dealing only with a subset of colored nodes -- explicitly colored ones -- as long as the requirement on explicit coloring is satisfied.
	}
	Lemma~\ref{lem:KN} %
	is naturally applied to the suffix tree if we define explicitly and implicitly $x$-colored nodes to be nodes with a hard $x$-link and a soft $x$-link, respectively. Lemma~\ref{lem:SREreversed}(\ref{case:SREreversed3}) ensures that the definition of implicit coloring is verified. It follows that to locate $\alpha$, soft W-links are not needed, and corresponding colorings can be omitted in the Euler tour. 
	This observation was used in  \cite{kucherov17fullfledged} to avoid storing soft links altogether and retrieve $\alpha$ based on hard links only, even if the $x$-link of $\alpha$ is soft. 
	
	Moreover, the authors of \cite[Thm.~1]{kucherov17fullfledged} showed that hard W-links are sufficient not only to locate $\alpha$ but also to recover the $x$-link $W_x(\alpha)$ needed to locate the insertion point $\gamma$ (Section~\ref{sec:weiner}). This is possible due to a stronger version of Lemma~\ref{lem:SREreversed}(\ref{case:SREreversed3}): not only a node $\nu$ with a soft $x$-link has a descendant with a hard $x$-link, but the unique closest such descendant has a hard $x$-link pointing to the same node as $W_x(\nu)$ (see also \cite[Observation 15(c)]{DBLP:journals/corr/abs-1302-3347}).
	As a consequence, maintaining {soft} W-links can be avoided altogether, which greatly simplifies the implementation of Weiner's algorithm. In particular, the work of setting or updating soft W-links (see Section~\ref{sec:weiner}) vanishes. We stick to this implementation in our work. %

	\subsection{Colored predecessor queries on lists}
	\label{sec:predecessor}
	
	As follows from the above, \textsc{colored predecessor} is the main tool for implementing Weiner's algorithm with worst-case guarantees.
	A data structure for solving this problem has been presented by Mortensen~\cite{mortensen06fully}.

	\begin{theorem}[\cite{mortensen06fully} Theorem 15]\label{thm:color-pred}
		For \textsc{colored predecessor}, there is an $O(n)$-space data structure supporting both updates and queries in worst-case time $O(\log^2\log n)$ where $n$ is the size of the list. 
		\label{thm:Mortensen}
		\end{theorem}

The bound of Theorem~\ref{thm:Mortensen} can be improved if the color space is log-size. 
\begin{theorem}[\cite{DBLP:journals/talg/GiyoraK09}, Theorem 5.2]\label{thm:color-pred-log}
	Assuming $|C|=O(\log^{1/4}n)$, \textsc{colored predecessor} can be solved in worst-case time $O(\log\log n)$ for both updates and queries, using an $O(n)$-space data structure. 
		\label{thm:GioraKaplan}
\end{theorem}

We remark that the result of \cite{DBLP:journals/talg/GiyoraK09} is slightly more general in that it allows querying a subset of colors and not just a single color.  On the other hand, in \cite{DBLP:journals/talg/GiyoraK09}, the $O(\log\log n)$ time bound for insertions and deletions is only amortized. In our case, we do not use deletions, and the insertion operation can be deamortized using standard techniques \cite{kucherov17fullfledged}. A similar result is proved in \cite[Theorem 4.1]{DBLP:conf/soda/Mortensen03} without providing a specific exponent in the bound on the color space size.

	Before concluding with the final complexity bound on a round of Weiner's algorithm, note that it also includes parent-to-child access in the suffix tree, however the complexity of this step is subsumed by \textsc{colored predecessor} (see \cite{KopplKucherov2026} for more details).
	\begin{theorem}
		When the string is processed online right-to-left, the suffix tree can be updated in $O(\log^2\log n)$ worst-case time per letter, for an integer alphabet. 
		This bound can be reduced to $O(\log\log n)$ for an alphabet size $\sigma=O(\log^{1/4}n)$. 
		\label{final-weiner}
		\end{theorem}

	\section{Near-real-time computation of smallest suffixient set}
	\label{sec:results}
	Based on efficient implementations of Weiner's algorithm, we now present our near-real-time solutions to the problem of maintaining the set $\sext(w)$ of supermaximal right-extensions (SREs) and of a smallest suffixient set of positions (SSS). We consider two frameworks: when the text is provided online right-to-left and when it is provided left-to-right. 	
	Both frameworks have been previously considered in \cite[Section~7]{DBLP:journals/corr/abs-2506-05638} 
	where the authors described online $O(n)$-time algorithms for both. Those algorithms, however, provide no worst-case time guarantee per individual letter, which is our primary goal in this work. 

	\subsection{Right-to-left framework}
	\label{sec:right-to-left}
	
	In this section, we assume that the text is processed right-to-left, starting from the unique terminal letter $\$$. Weiner's algorithm maintains the suffix tree online in this framework. Below we will show how Weiner's algorithm can be modified to maintain the set $\sext(w)$ and an SSS. 
	
	First, let us clarify how SREs are identified in Weiner's suffix tree. 
	A SRE $(u,x)$ corresponds to a node $\locus(u)$ with a descending $x$-edge (i.e. an edge whose label starts with $x$). The following lemma characterizes those edges. 
	
	\begin{lemma}
		Let $ux\in\fact_w$ for $u\in\Sigma^*$ and $x\in\Sigma$. 
		Then $(u,x)$ is a SRE if and only if 
		\textit{(i)} $\locus(u)$ is a node $\beta$,
		\textit{(ii)} one of $\beta$'s descending edges is an $x$-edge, 
		and \textit{(iii)} for every \textit{hard} W-link of $\beta$, say $W_z(\beta)=\nu$ for a $z \in \Sigma$, $\nu$ does not have a descending $x$-edge.
		\label{lem:SREedges}
	\end{lemma}
	\begin{proof}
		By the definition of a SRE, 
		$u$ must be right-maximal and one of its occurrences is followed by $x$, which implies (i) and (ii).  
		Furthermore, the definition requires that for every $z\in\Sigma$, 
		either $W_z(\locus(u))$ does not have a descending $x$-edge, or $W_z(\locus(u))$ is an implicit node which means that $W_z(\locus(u))$ is soft --- otherwise $u$ is a suffix of a longer right-maximal repeat.
		That is, the excluded case is when $W_z(\locus(u))$ is hard and has a descending $x$-edge; (iii) follows. 
	\end{proof}
	
	We now analyse how the set $\sext(w)$ of SREs can be modified when a letter $x\in\Sigma$ is prepended to $w$. 
	We analyse this in terms of a round of Weiner's algorithm described in Section~\ref{sec:weiner} and illustrated in Figure~\ref{fig:weiner_suffix_update}. Recall that $\strlabel(\alpha)$ is the longest prefix of $w$ that has an occurrence in $w$ preceded by $x$, and then $\strlabel(\gamma)=x\,\strlabel(\alpha)$ is the longest prefix of $xw$ that has another occurrence in $xw$. Note that $\strlabel(\alpha)$ is right-maximal in $w$ and $\strlabel(\gamma)$ is right-maximal in $xw$. 
	Let $y \in \Sigma$ such that $w$ starts with $\strlabel(\alpha) \cdot y$ 
	(and therefore $xw$ starts with $\strlabel(\gamma) \cdot y $). 

	We now use Lemma~\ref{lem:SREedges} to determine the changes in $\sext(w)$ caused by a round of Weiner's algorithm. Based on Lemma~\ref{lem:SREedges}, we identify a SRE $(u,x)$ with the pair $(\locus(u),x)$.
	By definition, $\locus(u)$ is a node that has a descending $x$-edge. 
	The following lemma specifies the modifications depending on whether $W_x(\alpha)$ is hard or soft.
	\begin{lemma}
			\begin{enumerate} [(a)]	
			\item If $W_x(\alpha)$ is hard, i.e. $\gamma$ is an existing node, then $(\gamma,y)$ is a new SRE of $xw$ that has to be added when transforming $\sext(w)$ to $\sext(xw)$. %
			Furthermore, if $(\alpha,y)$ was a SRE in $\sext(w)$, by Lemma~\ref{lem:SREedges} it is no longer a SRE in $\sext(xw)$ and should be omitted. 
			\item If $W_x(\alpha)$ is soft, then $(\gamma,y)$ and $(\gamma,z)$ are two new SREs of $xw$ associated with the new node $\gamma$ and its two descending edges (a $y$-edge and a $z$-edge). 
				Furthermore, if any of $(\alpha,y)$ or $(\alpha,z)$ is a SRE in $\sext(w)$, it is not a member of $\sext(xw)$. 
		\end{enumerate}
		\label{lem:righttoleft}
		\end{lemma}
		\begin{proof}
				\begin{enumerate} [(a)]	
					\item First observe that any hard W-link of $\gamma$ (if there is any) points to a node without descending $y$-edge. Indeed, by definition of $\alpha$, $\strlabel(\gamma)=x\cdot\strlabel(\alpha)$ does not  have right extension $y$ before the round, and the only occurrence of $\strlabel(\gamma)\cdot y$ after the round is a prefix occurrence. Therefore, for any $z\in\Sigma$, $z\cdot\strlabel(\gamma)\cdot y$ does not occur in $xw$ which means that $W_z(\gamma)$ cannot have a descending $y$-edge. 
					
					From the above, by Lemma~\ref{lem:SREedges}, a new descending $y$-edge created for $\gamma$ defines a new SRE $(\gamma,y)$. On the other hand, if $(\alpha,y)$ was a SRE, it is no longer one because $\gamma=W_x(\alpha)$ has now a descending $y$-edge. 
					\item Recall that the $y$-edge of $\gamma$ connects $\gamma$ to the newly created leaf ($\lambda'$ in Figure~\ref{fig:weiner_suffix_update}) and then the $z$-edge connects $\gamma$ to the previously existing descendant ($\beta$ in Figure~\ref{fig:weiner_suffix_update}). In terms of the text~$w$, 
						$y$ is the letter following the prefix occurrence of $\strlabel(\gamma)$ in $xw$ and $z$ is the unique letter following its other occurrences. 
						Also recall that $\gamma$ inherits the W-links of its already existing child ($\beta$ in the figure), but all these W-links become soft. 
						Indeed, if $W_t(\gamma)$  were hard for some $t\in\Sigma$, then $t \cdot \strlabel(\gamma)$ would be the label of either an internal node or a leaf before the round.
						In the former case, $t \cdot \strlabel(\gamma)$ would have two distinct right extensions, and hence $\strlabel(\gamma)$ would be right-maximal. 
						In the latter case,  $\strlabel(\gamma)$ would be a suffix of the current text.
						In both cases, $\gamma$ would already be explicit before the split, a contradiction.
					Hence, $(\gamma,y)$ and $(\gamma,z)$ are new SREs. Furthermore, similar to the previous case, if any of $(\alpha,y)$ and $(\alpha,z)$ was a SRE in $\sext(w)$, it is no longer a SRE in $xw$, i.e. it is not a member of $\sext(xw)$. 
				\end{enumerate}
		\end{proof}
	Thus, a round of Weiner's algorithm either introduces one new SRE and possibly deletes one, or introduces two new SREs and possibly deletes one or two. 
	
	We now turn to the implementation and show how to maintain the set of SREs and, at the same time, an SSS corresponding to this set. 
	As in \cite[before Theorem~3]{DBLP:journals/corr/abs-2506-05638}, it is convenient to think of positions as decrementing negative integers referring to the distance from the end of the string. At each node $\mu$ of the suffix tree, we store the \textit{end} position of the rightmost occurrence of $\strlabel(\mu)$. 
	For leaves, this position will be $0$. 
	For internal nodes, we maintain these positions in the standard way: 
	when a new internal node is created by splitting an edge (node $\gamma$ in Fig.~\ref{fig:weiner_suffix_update}), the rightmost end position of $\strlabel(\gamma)$ is computed relative to the descendant node ($\beta$), by subtracting the label length of the corresponding edge ($|\strlabel(\beta)|-|\strlabel(\gamma)|$). 
	
	Furthermore, each SRE $(u,x)$ is uniquely associated with an edge of the suffix tree (see Lemma~\ref{lem:SREedges}) and therefore can be associated with a unique rightmost end position of $ux$. The uniqueness is supported by the fact that no two SREs can end at the same position~\cite[Lemma~1]{DBLP:conf/spire/CenzatoOP24}.
	We compute this rightmost end position of a SRE $(u,x)$ from the position stored at the descendant node of the corresponding edge. 
	We then associate the current set of SREs to an SSS consisting of the rightmost end positions of each SRE. 
	To maintain this set, we use a left-growable bit vector storing $1$ at each bit position corresponding to the rightmost end position of one of the current SREs. 
	The growable vector can be implemented e.g. using the \emph{resizable array} of \cite{brodnik99resizable}. 
	When a SRE is inserted into or deleted from the current set, its rightmost end position is computed and the corresponding bit of the array is flipped. The size of the SSS is the number of 1's in the array. 
	
	We conclude with the final result of this section.
	\begin{theorem}
		In the right-to-left framework, the set of SREs and an SSS consisting of the rightmost positions of each SRE can be maintained online in worst-case $O(\timeSU)$ time per letter. This bound can be reduced to $O(\log\log n)$ for alphabet size $\sigma=O(\log^{1/4}n)$. 
	\end{theorem}

	\subsection{Left-to-right framework} 
	\label{sec:left-to-right}
	If we receive the letters of the text from left to right, we can use Ukkonen's algorithm to maintain the suffix tree online. 
	This algorithm, in turn, can be used to maintain online the set of SREs and an SSS~\cite[Section~7.1]{DBLP:journals/corr/abs-2506-05638}. Ukkonen's algorithm runs in linear time for constant-sized alphabets, however it provides no worst-case guarantee on the time spent on an individual letter. Here we apply the modified Weiner's algorithm from Section~\ref{sec:algotools} to obtain an efficient near-real-time algorithm for this problem. 
	
	\subsubsection{Modifications to the set of SREs}
	\label{sec:SREmodifications}
	
	We will be maintaining the suffix tree $\STree(\pali{w})$ of the reversal~$\pali{w}$ for an input text $w$ received online from left to right. 
	The reversal~$\pali{w}$ of $w[1..n]$ is the string $w[n] \cdot w[n-1] \cdots w[1]$.
	Consequently, $\STree(\pali{w})$ indexes the reverse~$\pali{u}$ of each substring $u\in\fact_w$.

	Analogous to Section~\ref{sec:right-to-left}, we first characterize how SREs of $w$ are represented in $\STree(\pali{w})$. 
	\begin{lemma}
		\label{lem:SRElefttoright}
		Let $ux\in \fact_w$ for $u\in \Sigma^*$ and $x\in \Sigma$. Then $(u,x)$ is a SRE in $w$ if and only if 
		(i) $\alpha = \locus(\pali{u})$ is a node in $\STree(\pali{w})$ with an $x$-link and at least another defined W-link, and 
		(ii) for each child node $\beta$ of $\alpha$, either $W_x(\beta)$ is not defined, or $W_x(\beta)$ is the only defined W-link of $\beta$. 
	\end{lemma}
	\begin{proof}
		Since $u$ is right-maximal in $w$, it has at least two right extensions in $w$, one of which is $x$, and therefore $\pali{u}$ is preceded by at least two distinct letters in $\pali{w}$, one of which is $x$. Thus, $\alpha=\locus(\pali{u})$ has at least two defined W-links in $\STree(\pali{w})$, one of which is an $x$-link. 
		By definition of SREs, for each $z\in\Sigma$, either $zux\notin\fact_w$ or $zu$ is followed only by $x$ in $w$. 
		This means that, for each $z \in \Sigma$ such that $\alpha$ has a $z$-edge, the child of $\alpha$ connected by this $z$-edge either has no $x$-link, or has an $x$-link as its sole W-link. 
		This also implies that $\alpha$ must be a node, as assuming that it has only one child contradicts Lemma~\ref{lem:SREreversed}(\ref{case:SREreversed2}) which implies that the W-links of an implicit node must be the same as those of its closest descendant. 
	\end{proof}
	
	Note that according to Lemma~\ref{lem:SRElefttoright}, here we are only concerned with what W-links are defined at a node and not with the destination of those W-links or with the distinction between hard and soft W-links. Therefore, for the purpose of this section, it is sufficient for us to assume that each node is associated with a subset of letters $z\in \Sigma$ with a defined $z$-link. 
	We now analyse how the set of SREs is modified when a new letter $x$ is appended to the current string $w$. Similar to Section~\ref{sec:right-to-left}, we infer these modifications from a round of Weiner's algorithm maintaining the suffix tree for the reversed string $\pali{w}$. 
	
	According to Lemma~\ref{lem:SRElefttoright}, we have to keep track of pairs $(\alpha,x)$, where $\alpha$ is a node and $x\in \Sigma$, such that $\alpha$ has an $x$-link as well as at least one W-link by another letter, and any child node of $\alpha$ either has only an $x$-link or does not have an $x$-link at all. 
	Like in the previous section, we will say that such a pair $(\alpha,x)$ is a SRE. 
	
	Appending $x$ to $w$ amounts to prepending $x$ to $\pali{w}$, and the corresponding updates to the suffix tree made by Weiner's algorithm have been described in Section~\ref{sec:weiner}. 
	In particular, recall that each new W-link belongs to one of the two types: 
	\begin{enumerate} [(i)]
		\item W-links defined for the newly created node $\gamma$ which are copied from the descendant node $\beta$, \label{item:gamma}
		\item $x$-links 
			defined for nodes that are on the path to the leaf $\lambda$ with $\strlabel(\lambda)=\pali{w}$. In Figure~\ref{fig:weiner_suffix_update}, these are nodes on the path between $\alpha$ (excluded) and $\lambda$ (included). \label{item:lambda}
	\end{enumerate}
	W-links of Type~(\ref{item:gamma}) are not relevant to us because the W-links of $\gamma$ are the same as the W-links of its descendant $\beta$, and therefore the conditions of Lemma~\ref{lem:SRElefttoright} cannot be satisfied. 
	Therefore, only W-links of Type~(\ref{item:lambda}) are of interest. This implies that if $(\mu,z)$ is a SRE to be added to or deleted from the current set of SREs, then $\mu$ is one of the nodes on the path between $\alpha$ (included) and $\lambda$ (excluded). For the other nodes of the tree, their set of W-links as well as the set of W-links of their children is not modified. 

	Based on Lemma~\ref{lem:SRElefttoright}, we now specify the possible modifications to the set of SREs that can occur at a round of Weiner's algorithm. As before, we assume that $x$ is the letter to be prepended to $\pali{w}$, $\lambda$ is the leaf labeled $\pali{w}$, and $\alpha$ is the node labeled by the longest prefix of $\pali{w}$ that has an occurrence in $\pali{w}$ preceded by $x$. 
Let $\nu$ be the parent of $\lambda$. 
The following %
lemma specifies all the updates that can occur. 
	\begin{lemma}
	\begin{enumerate} [(a)]
		\item If $\nu=\alpha$, that is, the parent of $\lambda$ has an $x$-link, then no modification has to be made to the set of SREs. \label{case:1}
				
		\hspace{-2.5em} The following cases assume that $\nu\neq \alpha$, i.e.\ $\nu$ is a proper descendant of $\alpha$. 
		
		\item $(\nu,x)$ is a new SRE. Furthermore, if $\nu$ had a single W-link by some letter $y\neq x$, then $(\nu,y)$ is a new SRE as well. \label{case:2}
		
		\item If $(\alpha,x)$ is a SRE, then it has to be omitted in $\sext(wx)$. \label{case:3}
		
		\item Let $\mu$ be the lowest proper ancestor of $\nu$ with two or more W-links such that its child on the path from $\mu$ to $\lambda$ has only one W-link by some letter $y\neq x$. If $(\mu,y)$ is a SRE, it has to be omitted in $\sext(wx)$. \label{case:4}
	\end{enumerate}
	\label{lem:lefttoright}
	\end{lemma}
	\begin{proof}
		Cases (\ref{case:2}), (\ref{case:3}) and (\ref{case:4}) are illustrated in \cref{fig:lrcases} in \cref{app:figs}. 
		\begin{enumerate} [(a)]
			\item In this case, Weiner's algorithm assigns a single new $x$-link to $\lambda$, which does not modify the set of SREs. 
			\item 
				By Lemma~\ref{lem:SREreversed}(\ref{case:SREreversed1}), $\nu$ had at least one W-link and could not have an $x$-link. Weiner's algorithm assigns an $x$-link to $\lambda$ and $\nu$ and therefore, by Lemma~\ref{lem:SRElefttoright}, $(\nu,x)$ becomes a SRE. 
				If, before the update, $\nu$ had exactly one W-link, say by letter $y$, then $(\nu,y)$ becomes a SRE as well, as all descendants of $\nu$ except for $\lambda$ must have a single W-link by $y$ and $\lambda$ has only the one by $x$. 
			\item Since $\alpha$ had an $x$-link and at least one of its children (the one on the path to $\lambda$) had no $x$-link but had one or more W-links by other letter(s), it may happen that $(\alpha,x)$ was a SRE. After the round of Weiner's algorithm, $(\alpha,x)$ is then no longer a SRE, as it has now a child with an $x$-link and at least one other W-link by another letter. 
			\item %
			Observe that such node $\mu$ always exists whenever $\nu$ has a single W-link, and in this case, $\mu$ is uniquely defined independently of $y$ as the lowest proper ancestor of $\nu$ with at least two W-links such that its child on the path to $\lambda$ has a single W-link. Then, if $(\mu,y)$ was a SRE, it ceases to be a SRE in $\sext(wx)$, as it has now a child with both a $y$-link and an $x$-link. 
		\end{enumerate}
\end{proof}

Before proceeding to algorithmic details, 
we show that the above cases are exhaustive.

	\begin{lemma}
		Lemma~\ref{lem:lefttoright} %
		specifies all possible modifications to the set of SREs.
	\end{lemma}
	\begin{proof}
		As noted earlier (W-links of Type~(\ref{item:lambda})), only nodes $\mu$ between $\alpha$ (included) and $\lambda$ (excluded) can form SREs to be added or deleted. 
		\hili{For the letter $x$, a new SRE $(\mu,x)$ can only occur if $\mu=\nu$, and is covered by Lemma~\ref{lem:lefttoright}(\ref{case:2}) ---
			for each other node~$\xi$ on this path, one of $\xi$'s children has a W-link by another letter (due to Lemma~\ref{lem:SREreversed}(\ref{case:SREreversed1})). 
			A SRE $(\mu,x)$ to be deleted can only occur if $\mu=\alpha$, and is covered by Lemma~\ref{lem:lefttoright}(\ref{case:3}). 
			A letter $y\neq x$ can cause a modification in two cases as well: 
			A new SRE $(\mu,y)$ can only appear for $\mu=\nu$ (Lemma~\ref{lem:lefttoright}(\ref{case:2})), 
			and a SRE $(\mu,y)$ can cease to be a SRE only if $\mu$ is the parent of a node having a single $y$-link, as specified in Lemma~\ref{lem:lefttoright}(\ref{case:4}). 
		} 
	\end{proof}
	
	Unsurprisingly, modifications described in Lemma~\ref{lem:lefttoright}(\ref{case:2}),(\ref{case:3}),(\ref{case:4}) correspond to those described in \cite[Section 7.1]{DBLP:journals/corr/abs-2506-05638}. The difference is that in  \cite{DBLP:journals/corr/abs-2506-05638}, modifications are described in terms of the suffix tree for $w$ maintained by Ukkonen's algorithm, whereas we describe them in terms of the suffix tree for $\pali{w}$ maintained by Weiner's algorithm. Below we show how to leverage our version in order to maintain the set of SREs in near-real-time. 
	
	Observe also that if the conditions of Lemma~\ref{lem:lefttoright}(\ref{case:4}) are satisfied, then the parent of $\lambda$ must have a single W-link, and therefore two new SREs are added according to \cref{lem:lefttoright}(\ref{case:2}). This implies that possible modifications amount to adding between zero and two new SREs~\cite[Lemma 2]{DBLP:journals/corr/abs-2506-05638}. 
	
	\subsubsection{Maintaining the set of SREs in near-real-time}\label{sec:maintainingSREs}

	Based on the results of Section~\ref{sec:SREmodifications}, we now describe how we can maintain the set of SREs in near-real-time relying on a near-real-time implementation of Weiner's algorithm from Section~\ref{sec:algotools}. %
	
	To compute the SREs to be added or deleted as specified by Lemma~\ref{lem:lefttoright}(\ref{case:2}),(\ref{case:3}),(\ref{case:4}),
	we have to support the following two operations on the suffix tree.
	\begin{enumerate} [(A)]
		\item given a leaf node $\lambda$ and a letter $x\in\Sigma$, find the lowest ancestor $\alpha$ of $\lambda$ with a defined $x$-link, \label{case:alpha}
		\item given a leaf node $\lambda$, find the lowest  ancestor $\mu$ of $\lambda$ with at least two defined W-links such that its child on the path to $\lambda$ has only one defined W-link. %
		 \label{case:mu}
	\end{enumerate}
	Operation~(\ref{case:alpha}) implements the update described in Lemma~\ref{lem:lefttoright}(\ref{case:3}), 
	whereas Operation~(\ref{case:mu}) implements the update described in Lemma~\ref{lem:lefttoright}(\ref{case:4}).
	Checking Lemma~\ref{lem:lefttoright}(\ref{case:1}) and implementing Lemma~\ref{lem:lefttoright}(\ref{case:2}) is trivial because $\nu$ is directly accessed  from its child $\lambda$. Note also that if $\nu$ has two or more W-links, Lemma~\ref{lem:lefttoright}(\ref{case:4}) does not apply. 

	Operation (\ref{case:alpha}) is a part of Weiner's algorithm and has been discussed in Section~\ref{sec:algotools} (see Fig.~\ref{fig:weiner_suffix_update}). We implement it using the reduction to \textsc{colored predecessor}, as explained in Sections~\ref{sec:fringe} and~\ref{sec:predecessor}. 
	
	Operation (\ref{case:mu}) is implemented similarly. 
	We use Lemma~\ref{lem:KN} with a single color: the colored nodes are nodes with two or more defined W-links. Let $M$ be the set of these nodes. By Lemma~\ref{lem:SREreversed}(\ref{case:SREreversed2}), $M$ is closed under ancestry relation, i.e.\ the coloring has the fringe property. 
	Our goal is to maintain some subset $D\subseteq M$ such that every node from $M$ has a descendant node from $D$. 
	Nodes of $D$ will be considered {explicitly colored} in the sense of Lemma~\ref{lem:KN}. 
	
	We reason by induction on the rounds of Weiner's algorithm. 
	Assume we have correctly computed the set $D$, that is every node with two or more W-links has an explicitly colored descendant. We show how to maintain $D$ through a round of Weiner's algorithm. We start by observing that the W-links assigned to the newly created node $\gamma$ (see Figure~\ref{fig:weiner_suffix_update}) do not cause any modifications to the set $D$ because the set of W-links of $\gamma$ is the same as that of its child $\beta$. Then, only the new W-links on the path between $\alpha$ (excluded) and $\lambda$ (included) may cause  modification. For the nodes on this path, Weiner's algorithm adds one new $x$-link ($x$ letter prepended at this round) 
	and therefore by Lemma~\ref{lem:SREreversed}(\ref{case:SREreversed1}), all these nodes except for $\lambda$ belong to $M$ after this round. Thus, it is sufficient to explicitly color the immediate parent of $\lambda$ unless it was already in $D$ before the round. In other words, if before the round the immediate parent of $\lambda$ had only one defined W-link, it has to be explicitly colored. All other nodes that are not on the path between $\alpha$ and $\lambda$ do not require any updates of $D$. We summarize the arguments in the following statement. 
	\begin{lemma}
		Explicitly coloring the immediate parent of $\lambda$ unless it was explicitly colored before correctly maintains the set $D$ of explicitly colored nodes. 
		\label{lem:explicitmark}
	\end{lemma}
	
	Since the coloring of Lemma~\ref{lem:explicitmark} is done by two updates of an Euler tour list with only one color, 
	Lemma~\ref{lem:KN} implies that Operation (\ref{case:mu}) can be implemented via \textsc{colored predecessor} query in the case of one color. 
	Thus, Operation~(\ref{case:mu}) can be executed in $O(\log \log n)$ worst-case time (cf.\ comment after Theorem~\ref{thm:GioraKaplan}). 
	
	Operation~(\ref{case:mu}) retrieves node $\mu$. 
	To retrieve the potential SRE $(\mu,y)$ (see Lemma~\ref{lem:lefttoright}(\ref{case:4})), 
	we need to determine the letter $y$ that corresponds to the only W-link of the child~$\eta$ of $\mu$ on the path to $\lambda$. 
	To determine the edge leading to $\eta$, we retrieve the string depth\footnote{String depths of nodes can be determined and maintained along Weiner's algorithm at the creation time of the nodes in constant time.}~$d$ of $\mu$ and then retrieve 
	$z := \pali{w}[d+1]$ (before prepending the new character $x$ to $\pali{w}$). 
	Then $y$ is retrieved from the W-link of $\eta$. 
	Since we do not store soft W-links explicitly, it may happen that the desired W-link is not stored at~$\eta$.
	To cope with this, we also maintain, 
	for every node with exactly one defined W-link, the corresponding letter of this W-link; 
	hence, once $\eta$ is found, its unique W-link letter gives $y$.
	Note that to ensure that the letter is correct, we do not need to update or remove the assigned W-link letter of a node along the algorithm, as we only query nodes that must have exactly one W-link.

	We now specify how we keep track of the current set of SREs and of an SSS. 
	Before executing the W-link updates of the round, 
	we determine which SREs have to be added or deleted, 
	as specified in Lemma~\ref{lem:lefttoright}(\ref{case:2}),(\ref{case:3}), and (\ref{case:4}).
	The corresponding bit-vector updates are then performed once the insertion locus and the relevant new W-link destinations are known.
	Similarly to Section~\ref{sec:right-to-left}, we keep track of the leftmost end positions of the current SREs in the current string $w$, exploiting the distinctness of this position for each SRE. 
	For that, we maintain a growable (to the right) vector of $O(n)$ bits marking those positions. 
	Since we work with the suffix tree of the reversed string $\pali{w}$, we maintain, at each node $\nu$, the \textit{rightmost start} position of $\strlabel(\nu)$ in the current string $\pali{w}$. This is done similarly to how it is described in Section~\ref{sec:right-to-left}. 
	It remains to specify how we compute the leftmost end position of a given SRE that we have to add to or delete from the current set. That is, if $(\nu,x)$ is a SRE, we need to compute the rightmost start position of $x \cdot \strlabel(\nu)$ in $\pali{w}$. 
	If $W_x(\nu)$ is hard, this position is retrieved directly from the destination node. 
	If $W_x(\nu)$ is soft, we use the technique of \cite{kucherov17fullfledged} described at the end of Section~\ref{sec:fringe} that enables the recovery of soft W-links without explicitly storing them. 
	Then, again, we retrieve the start position from the destination node.  
	
	The results of this section are thus summarized as follows. 
\begin{theorem}
	In the left-to-right framework, the set of SREs and the SSS consisting of the leftmost positions of each SRE can be maintained online in worst-case $O(\timeSU)$ time per letter. This bound can be reduced to $O(\log\log n)$ for alphabet size $\sigma=O(\log^{1/4}n)$. 
\end{theorem}

\section{Conclusion}
In this paper, we have shown how to maintain the set of SREs and an SSS in near-real-time when the text is received online from either right to left or left to right.
In both cases, we managed to leverage the near-real-time implementation of Weiner's algorithm to achieve the same time complexity bounds, i.e.\ 
worst-case $O(\timeSU)$ time per letter, or $O(\log\log n)$ for alphabet size $\sigma=O(\log^{1/4}n)$.
Further improvements of the time complexity bounds for maintaining the suffix tree with Weiner's algorithm in near-real-time would directly imply improvements for maintaining the set of SREs and an SSS.
While our space usage is $O(n)$,
it would be interesting to study whether we can leverage techniques based on the online construction of the Burrows--Wheeler transform (BWT) to maintain the set of SREs online using compact space or compressed space in terms of the number of letter runs in the BWT.

	\clearpage
\bibliographystyle{plain}

\begin{thebibliography}{10}
	
	\bibitem{breslauer13near}
	Dany Breslauer and Giuseppe~F. Italiano.
	\newblock Near real-time suffix tree construction via the fringe marked
	ancestor problem.
	\newblock {\em J. Discrete Algorithms}, 18:32--48, 2013.
	
	\bibitem{brodnik99resizable}
	Andrej Brodnik, Svante Carlsson, Erik~D. Demaine, J.~Ian Munro, and Robert
	Sedgewick.
	\newblock Resizable arrays in optimal time and space.
	\newblock In {\em Proc.\ WADS}, volume 1663 of {\em LNCS}, pages 37--48, 1999.
	
	\bibitem{DBLP:journals/corr/abs-2407-18753}
	Davide Cenzato, Lore Depuydt, Travis Gagie, Sung-Hwan Kim, Giovanni Manzini,
	Francisco Olivares, and Nicola Prezza.
	\newblock Suffixient arrays: a new efficient suffix array compression
	technique.
	\newblock {\em CoRR}, abs/2407.18753, 2025.
	
	\bibitem{DBLP:conf/spire/CenzatoOP24}
	Davide Cenzato, Francisco Olivares, and Nicola Prezza.
	\newblock On computing the smallest suffixient set.
	\newblock In Zsuzsanna Lipt{\'{a}}k, Edleno~Silva de~Moura, Karina Figueroa,
	and Ricardo Baeza{-}Yates, editors, {\em String Processing and Information
		Retrieval - 31st International Symposium, {SPIRE} 2024, Puerto Vallarta,
		Mexico, September 23-25, 2024, Proceedings}, volume 14899 of {\em Lecture
		Notes in Computer Science}, pages 73--87. Springer, 2024.
	
	\bibitem{DBLP:journals/corr/abs-2506-08225}
	Davide Cenzato, Francisco Olivares, and Nicola Prezza.
	\newblock Testing suffixient sets.
	\newblock {\em CoRR}, abs/2506.08225, 2025.
	
	\bibitem{10.1145/3426473}
	Anders~Roy Christiansen, Mikko~Berggren Ettienne, Tomasz Kociumaka, Gonzalo
	Navarro, and Nicola Prezza.
	\newblock Optimal-time dictionary-compressed indexes.
	\newblock {\em ACM Trans. Algorithms}, 17(1), December 2021.
	
	\bibitem{cole05dynamic}
	Richard Cole and Ramesh Hariharan.
	\newblock Dynamic {LCA} queries on trees.
	\newblock {\em {SIAM} J. Comput.}, 34(4):894--923, 2005.
	
	\bibitem{DBLP:journals/corr/abs-2312-01359}
	Lore Depuydt, Travis Gagie, Ben Langmead, Giovanni Manzini, and Nicola Prezza.
	\newblock Suffixient sets.
	\newblock {\em CoRR}, abs/2312.01359, 2023.
	
	\bibitem{DBLP:journals/corr/abs-1302-3347}
	Johannes Fischer and Pawel Gawrychowski.
	\newblock Alphabet-dependent string searching with wexponential search trees.
	\newblock {\em CoRR}, abs/1302.3347, 2013.
	
	\bibitem{DBLP:journals/corr/abs-2506-05638}
	Hiroto Fujimaru, Gonzalo Navarro, Giuseppe Romana, and Cristian Urbina.
	\newblock Smallest suffixient sets: Effectiveness, resilience, and calculation.
	\newblock {\em CoRR}, abs/2506.05638, 2025.
	
	\bibitem{DBLP:journals/talg/GiyoraK09}
	Yoav Giyora and Haim Kaplan.
	\newblock Optimal dynamic vertical ray shooting in rectilinear planar
	subdivisions.
	\newblock {\em {ACM} Trans. Algorithms}, 5(3):28:1--28:51, 2009.
	
	\bibitem{gusfield97algorithms}
	Dan Gusfield.
	\newblock {\em Algorithms on Strings, Trees, and Sequences: Computer Science
		and Computational Biology}.
	\newblock Cambridge University Press, 1997.
	
	\bibitem{kempa18stringattractors}
	Dominik Kempa and Nicola Prezza.
	\newblock At the roots of dictionary compression: string attractors.
	\newblock In {\em Proc.\ STOC}, pages 827--840, 2018.
	
	\bibitem{KopplKucherov2026}
	Dominik K\"oppl and Gregory Kucherov.
	\newblock Near-real-time solutions for online string problems.
	\newblock {\em CoRR}, abs/2602.15311, 2026.
	\newblock to appear in CPM 2026.
	
	\bibitem{kucherov17fullfledged}
	Gregory Kucherov and Yakov Nekrich.
	\newblock Full-fledged real-time indexing for constant size alphabets.
	\newblock {\em Algorithmica}, 79(2):387--400, 2017.
	
	\bibitem{DBLP:journals/corr/abs-2510-16454}
	Gregory Kucherov and Yakov Nekrich.
	\newblock Online computation of normalized substring complexity.
	\newblock {\em CoRR}, abs/2510.16454, 2025.
	\newblock appeared in LATIN 2026.
	
	\bibitem{DBLP:conf/soda/Mortensen03}
	Christian~Worm Mortensen.
	\newblock Fully-dynamic two dimensional orthogonal range and line segment
	intersection reporting in logarithmic time.
	\newblock In {\em Proceedings of the Fourteenth Annual {ACM-SIAM} Symposium on
		Discrete Algorithms, January 12-14, 2003, Baltimore, Maryland, {USA}}, pages
	618--627. {ACM/SIAM}, 2003.
	
	\bibitem{mortensen06fully}
	Christian~Worm Mortensen.
	\newblock Fully dynamic orthogonal range reporting on {RAM}.
	\newblock {\em {SIAM} J. Comput.}, 35(6):1494--1525, 2006.
	
	\bibitem{navarro21indexing1}
	Gonzalo Navarro.
	\newblock Indexing highly repetitive string collections, part {I:}
	repetitiveness measures.
	\newblock {\em {ACM} Comput. Surv.}, 54(2):29:1--29:31, 2021.
	
	\bibitem{navarro21indexing2}
	Gonzalo Navarro.
	\newblock Indexing highly repetitive string collections, part {II:} compressed
	indexes.
	\newblock {\em {ACM} Comput. Surv.}, 54(2):26:1--26:32, 2021.
	
	\bibitem{DBLP:conf/spire/NavarroRU25}
	Gonzalo Navarro, Giuseppe Romana, and Cristian Urbina.
	\newblock Smallest suffixient sets as a repetitiveness measure.
	\newblock In Golnaz Badkobeh, Jakub Radoszewski, Nicola Tonellotto, and Ricardo
	Baeza{-}Yates, editors, {\em String Processing and Information Retrieval -
		32nd International Symposium, {SPIRE} 2025, London, UK, September 8-11, 2025,
		Proceedings}, volume 16073 of {\em Lecture Notes in Computer Science}, pages
	217--232. Springer, 2025.
	
	\bibitem{raskhodnikova13sublinear}
	Sofya Raskhodnikova, Dana Ron, Ronitt Rubinfeld, and Adam~D. Smith.
	\newblock Sublinear algorithms for approximating string compressibility.
	\newblock {\em Algorithmica}, 65(3):685--709, 2013.
	
	\bibitem{ShibataBannai26}
	Hiroki Shibata and Hideo Bannai.
	\newblock String representation in suffixient set size space.
	\newblock {\em CoRR}, abs/2604.04377, 2026.
	
	\bibitem{weiner73linear}
	Peter Weiner.
	\newblock Linear pattern matching algorithms.
	\newblock In {\em Proc. 14th Symposium on Switching and Automata Theory
		(SWAT)}, pages 1--11, 1973.
	
\end{thebibliography}

\clearpage
\appendix

\section{Additional figures}\label{app:figs}

\begin{figure}[h]
	\centering
	\begin{minipage}{0.45\linewidth}
	\includegraphics[width=\linewidth]{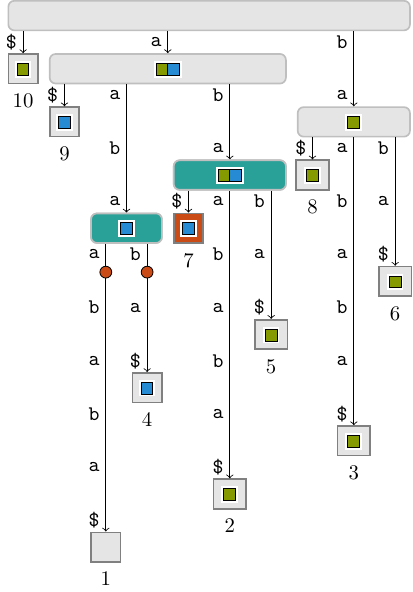}

	\setlength{\tabcolsep}{2pt}
	Suffix tree of 
\begin{tabular}{r*{10}{c}}
$i = $ & 1 & 2 & 3 & 4 & 5 & 6 & 7 & 8 & 9 & 10 \\
$w[i] =$ & \texttt{a} & \texttt{a} & \texttt{b} & \texttt{a} & \underline{\texttt{a}} & \texttt{b} & \texttt{a} & \underline{\texttt{b}} & \texttt{a} & \underline{\texttt{\$}} \\
\end{tabular}
	\end{minipage}
	\hfill
	\begin{minipage}{0.45\linewidth}
	\includegraphics[width=\linewidth]{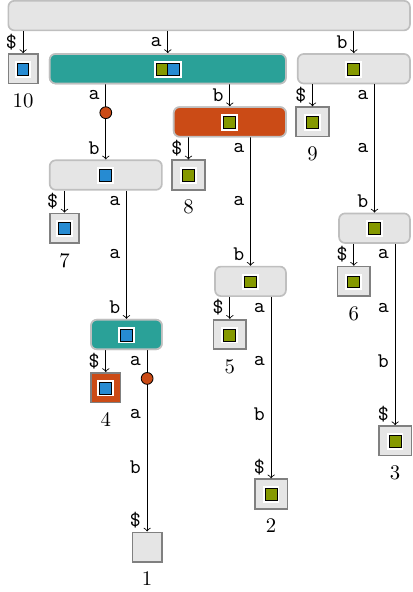}

	\setlength{\tabcolsep}{2pt}
	Suffix tree of 
\begin{tabular}{r*{10}{c}}
$i = $ & 1 & 2 & 3 & 4 & 5 & 6 & 7 & 8 & 9 & 10 \\
$v[i] =$ & \texttt{a} & \underline{\texttt{a}} & \underline{\texttt{b}} & \texttt{a} & \texttt{a} & \texttt{b} & \underline{\texttt{a}} & \texttt{a} & \texttt{b} & \underline{\texttt{\$}} \\
\end{tabular}
	\end{minipage}
	\caption{Suffix trees of the strings in Example~\ref{ex:suffixient_example}. 
		The locus $ux$ for each SRE~$(u,x)$ marked by an orange dot ({\color{solarizedOrange}$\blacksquare$})
	can be either a node or an implicit node; 
	the node with label $u$ is further filled in cyan ({\color{solarizedCyan}$\blacksquare$}).
	Leaves are labeled below by the starting position of the corresponding suffix in the string, e.g.\ the leaf labeled by $1$ corresponds to the entire string.
	Nodes are augmented with the colored boxes similar to Fig.~\ref{fig:weiner_suffix_update} to visualize the presence of an $\mathtt{a}$-link with green ({\color{solarizedGreen}$\blacksquare$}) and a $\mathtt{b}$-link with blue ({\color{solarizedBlue}$\blacksquare$}).
	\textbf{Left:} $w$ has only one SSS since all orange marked loci are leaves or on edges connecting to leaves. This SSS is depicted on bottom by underlining the letters at the end positions of the SREs. 
	\textbf{Right:} Since $\mathtt{aa}$ and $\mathtt{ab}$ each have three occurrences in $v$ (by counting the number of leaves in the respective subtrees), $v$ has nine SSSs. They are depicted separately in Figure~\ref{fig:sss_full_example}. 
	Each SSS can be obtained by picking, for each SRE~$(u,x)$, a suffix number of a leaf of the subtree of its locus (marked in orange) incremented with $|u|$.
}
	\label{fig:suffixient_example}
\end{figure}

\begin{figure}[t]
	\centering
	\includegraphics{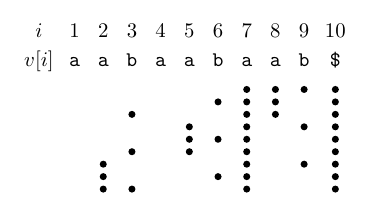}
	\caption{All SSSs of $v = \mathtt{aabaabaab\$}$ depicted by rows of dots.
		By \cref{fig:suffixient_example}, the positions $7$ and $10$ are present in all SSS since their loci are leaves or on edges leading to leaves (leaves with suffix numbers 1 and 4 in the figure).
	}
	\label{fig:sss_full_example}
\end{figure}

\begin{figure}[h]
	\centering
	\includegraphics[valign=m,page=1]{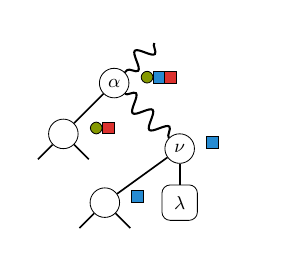}
	\quad $\longrightarrow$ \quad
	\includegraphics[valign=m,page=2]{./suffixientimg/lrcases.pdf}

	Lemma~\ref{lem:lefttoright}(\ref{case:2}) 

	\includegraphics[valign=m,page=3]{./suffixientimg/lrcases.pdf}
	\quad $\longrightarrow$ \quad
	\includegraphics[valign=m,page=4]{./suffixientimg/lrcases.pdf}

	Lemma~\ref{lem:lefttoright}(\ref{case:3})

	\includegraphics[valign=m,width=0.45\textwidth,page=5]{./suffixientimg/lrcases.pdf}
	$\longrightarrow$ 
	\includegraphics[valign=m,width=0.45\textwidth,page=6]{./suffixientimg/lrcases.pdf}

	Lemma~\ref{lem:lefttoright}(\ref{case:4})

	\caption{Illustration to the modifications of the set of SREs specified in Lemma~\ref{lem:lefttoright}(\ref{case:2}),(\ref{case:3}),(\ref{case:4}) of \cref{sec:SREmodifications}.
		Each subfigure depicts a transformation resulting from a round of Weiner's algorithm. 
		\hili{Like in Figure~\ref{fig:weiner_suffix_update}, squares represent individual letters and circles denote sets of letters. 
			Color-filled nodes represent SREs: a node $\chi$ colored $y$ witnesses a SRE $(\zeta,y)$, where $\chi$ is a child of $\zeta$ having only a $y$-link in its set of W-links.
	Lemma~\ref{lem:lefttoright}(\ref{case:2}): the subfigure depicts the case when $\nu$ has a single W-link ({\color{solarizedBlue}$\blacksquare$}) before the round, in which case $\nu$ induces a new blue SRE in addition to the new red SRE. If $\nu$ has two or more W-links before the round, only the red SRE is created. %
	Lemma~\ref{lem:lefttoright}(\ref{case:3}): the subfigure depicts the case when $\alpha$ formed a red SRE before the round, which ceases to be a SRE after the round. 
	Lemma~\ref{lem:lefttoright}(\ref{case:4}): $\nu$ has a single color and $\mu$ is the lowest ancestor of $\nu$ with two or more colors. The blue SRE ceases to be a SRE. 
}
	}
	\label{fig:lrcases}
\end{figure}

\end{document}